\documentclass{amsart}

\usepackage{enumerate, amsmath, amsfonts, amssymb, amsthm, wasysym, graphics, graphicx, xcolor, url, hyperref, hypcap, multirow, a4wide}
\hypersetup{colorlinks=true, citecolor=darkblue, linkcolor=darkblue}
\usepackage[bottom]{footmisc}
\graphicspath{{figures/}}

\title[Quasi-configurations: Building Blocks for Point\,--\,Line Configurations]{Quasi-configurations: \\ Building Blocks for Point\,--\,Line Configurations}

\thanks{VP was partially supported by the spanish MICINN grant MTM2011-22792 and by the French ANR grant EGOS 12 JS02 002 01.}

\author{J\"urgen Bokowski}
\address[JB]{Technische Universit\"at Darmstadt}
\email{juergen.bokowski@gmail.com}
\urladdr{http://wwwopt.mathematik.tu-darmstadt.de/~bokowski/}

\author{Vincent Pilaud}
\address[VP]{CNRS \& LIX, \'Ecole Polytechnique, Palaiseau}
\email{vincent.pilaud@lix.polytechnique.fr}
\urladdr{http://www.lix.polytechnique.fr/~pilaud/}

\newtheorem{result}{Result}
\newtheorem{proposition}[result]{Proposition}
\newtheorem{corollary}[result]{Corollary}

\newtheorem{question}[result]{Question}

\theoremstyle{definition}
\newtheorem{example}[result]{Example}

\newcommand{\N}{\mathbb{N}}

\newcommand{\bbP}{\mathbb{P}}

\newcommand{\bP}{\mathbf{P}}
\newcommand{\bL}{\mathbf{L}}

\newcommand{\set}[2]{\left\{ #1 \;\middle|\; #2 \right\}}

\newcommand{\eqdef}{\mbox{\,\raisebox{0.2ex}{\scriptsize\ensuremath{\mathrm:}}\ensuremath{=}\,}} 

\newcommand{\ie}{\textit{i.e.}} 

\definecolor{darkblue}{rgb}{0,0,0.7}
\newcommand{\darkblue}{\color{darkblue}}
\newcommand{\defn}[1]{\emph{\darkblue #1}}
\graphicspath{{figures/}}
\renewcommand{\paragraph}[1]{\vspace{.3cm}\noindent{\sc #1} --- }

\usepackage{todonotes}

\begin{document}

\begin{abstract}
We study point\,--\,line incidence structures and their properties in the projective plane. Our motivation is the problem of the existence of $(n_4)$ configurations, still open for few remaining values of~$n$. Our approach is based on quasi-configurations: point\,--\,line incidence structures where each point is incident to at least~$3$ lines and each line is incident to at least~$3$ points. We investigate the existence problem for these quasi-configurations, with a particular attention to~$3|4$-configurations where each element is $3$- or $4$-valent. We use these quasi-configurations to construct the first $(37_4)$ and $(43_4)$ configurations. The existence problem of finding $(22_4)$, $(23_4)$, and~$(26_4)$ configurations remains open.


\medskip
\noindent
\textsc{MSC Classes.} 52C30
\end{abstract}

\vspace*{-1.1cm}

\maketitle

\vspace{-.5cm}

\section{Introduction}
\label{sec:introduction}

\enlargethispage{.3cm}
A \defn{geometric $(n_k)$ configuration} is a collection of~$n$ points and~$n$ lines in the projective plane such that each point lies on~$k$ lines and each line contains~$k$ points. We recommend our reader to consult Gr\"unbaum's book~\cite{Grunbaum1} for a comprehensive presentation and an historical perspective on these configurations. The central problem studied in this book is to determine for a given $k$ those numbers $n$ for which there exist geometric $(n_k)$ configurations. The answer is completely known for $k = 3$ (geometric $(n_3)$ configurations exist if and only if~$n \ge 9$), partially solved for~$n = 4$ (geometric $(n_4)$ configurations exist if and only if~$n = 18$ or~$n \ge 20$ with a finite list of possible exceptions), and wide open for~$k > 4$. Our contribution concerns~$k = 4$, where we provide solutions for two former open cases: there exist geometric $(37_4)$ and $(43_4)$ configurations. Moreover, we study building blocks for constructing geometric $(n_k)$ configurations that might be of some help for clarifying the final open cases $(22_4)$, $(23_4)$, and~$(26_4)$.  Many aspects of our presentation appeared during our investigation of the case $(19_4)$ in which there is no geometric $(19_4)$ configuration, see~\cite{BokowskiPilaud, BokowskiPilaud2}. 

The approach of this paper is to construct geometric $(n_4)$ configurations from smaller building blocks. For example, Gr\"unbaum's geometric $(20_4)$ configuration~\cite{Grunbaum5} can be constructed by superposition of two geometric $(10_3)$ configurations as illustrated in Figure~\ref{fig:splittings20_4}. To extend this kind of construction, we study an extended notion of point\,--\,line configurations, where incidences are not regular but still prescribed.

\begin{figure}
  \centerline{\includegraphics[width=.7\textwidth]{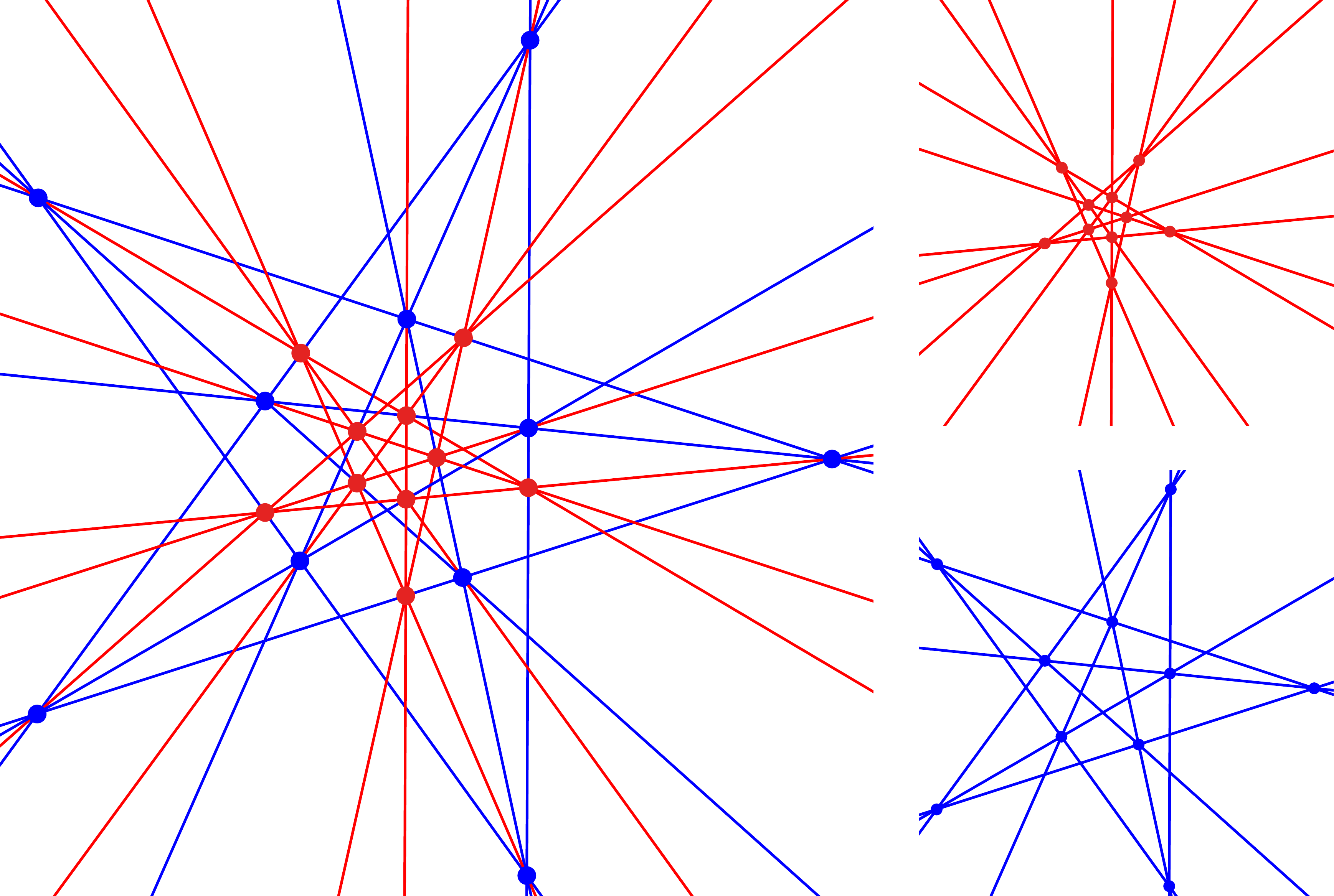}}
  \caption{Splitting Gr\"unbaum's geometric $(20_4)$ configuration~\cite{Grunbaum5} into two $(10_3)$ configurations.}
  \label{fig:splittings20_4}
\end{figure}


\subsection{Point\,--\,line incidence structures}

We define a \defn{point\,--\,line incidence structure} as a set~$P$ of \defn{points} and a set~$L$ of \defn{lines} together with a point\,--\,line \defn{incidence relation}, where two points of~$P$ can be incident with at most one line of~$L$ and two lines of~$L$ can be incident with at most one point of~$P$. Throughout the paper, we only consider \defn{connected} incidence structures, where any two elements of~$P \sqcup L$ are connected via a path of incident elements.

For a point\,--\,line incident structure~$(P,L)$, we denote by~$p_i$ the number of points of~$P$ contained in~$i$ lines of~$L$ and similarly by~$\ell_j$ the number of lines of~$L$ containing~$j$ points of~$P$. We find it convenient to encode these incidence numbers into a pair of polynomials~$(\bP(x), \bL(y))$, called the \defn{signature} of~$(P,L)$, and defined by
\[
\bP(x) \eqdef \sum\nolimits_i p_i x^i \qquad \text{and} \qquad \bL(y) \eqdef \sum\nolimits_j \ell_j y^j.
\]
For example, Figure~\ref{fig:8282} shows a point\,--\,line incidence structure with signature~$(8x^3+2x^4, 8y^3+2y^4)$. With these notations, the number of points and lines are given by~$|P| = \bP(1)$ and~$|L| = \bL(1)$, and the number of point\,--\,line incidences is~$|\set{(p,\ell) \in P \times L}{p \in \ell}| = \bP'(1) = \bL'(1)$.

We distinguish three different levels of point\,--\,line incidence structures, in increasing generality:
\begin{description}
\item[\rm \defn{Geometric}] Points and lines are ordinary points and lines in the real projective plane~$\bbP$.
\item[\rm \defn{Topological}] Points are ordinary points in~$\bbP$, but lines are \defn{pseudolines}, \ie, non-separating simple closed curves of~$\bbP$ which cross pairwise precisely once.
\item[\rm \defn{Combinatorial}] Just an abstract incidence structure $(P,L)$ as described above, with no additional geometric structure.  
\end{description}
In this paper, we are mainly interested in the geometric level. We therefore omit the word geometric in what follows unless we have to distinguish different levels.


\subsection{$(n_k)$ configurations}

One of the main problems in the theory of point\,--\,line incidence structures is to clarify the existence of regular point\,--\,line incidence structures. A \defn{$k$-configuration} is a point\,--\,line incidence structure~$(P,L)$ where each point of~$P$ is contained in~$k$ lines of~$L$ and each line of~$L$ contains~$k$ points of~$P$. In such a configuration, the number of points equals the number of lines, and thus it has signature~$(n x^k, n y^k)$. If we want to specify the number of points and lines, we call it an \defn{$(n_k)$ configuration}. We refer to the recent monographs of Gr\"unbaum~\cite{Grunbaum1} and Pisanski and Servatius~\cite{PisanskiServatius} for comprehensive presentations of these objects. Classical examples of regular configurations are Pappus' and Desargues' configurations, which are respectively $(9_3)$ and $(10_3)$ configurations. In the study of the existence of $(n_4)$ configurations there are still a few open cases. Namely, it is known that (geometric) $(n_4)$-configurations exists if and only if~$n = 18$ or~$n \ge 20$, with the possible exceptions of~$n = 22, 23, 26, 37$ and~$43$ \cite{Grunbaum2, BokowskiSchewe2, BokowskiPilaud2}. Different methods have been used to obtain the current results on the existence of $4$-configurations:
\begin{enumerate}[(i)]
\item For~$n \le 16$, Bokowski and Schewe~\cite{BokowskiSchewe1} used a counting argument based on Euler's formula to prove that there exist no $(n_4)$ configuration, even topological.
\item For small values of~$n$, one can search for all possible $(n_4)$ configurations. For~$n=17$ or~$18$, one can first enumerate all combinatorial $(n_4)$ configurations and search for geometric realizations among them. This approach was used by Bokowski and Schewe~\cite{BokowskiSchewe2} to show that there is no $(17_4)$ configuration and to produce a first $(18_4)$ configuration. Another approach, proposed in~\cite{BokowskiPilaud}, is to enumerate directly all topological $(n_4)$ configurations, and to search for geometric realizations among this restricted family. In this way, we showed that there are precisely two $(18_4)$ configurations, that of~\cite{BokowskiSchewe2} and another one~\cite{BokowskiPilaud}, see Figure~\ref{fig:splittings18_4}. For~$n=19$, we obtained in~\cite{BokowskiPilaud} all $4\,028$ topological $(19_4)$ configurations and the study of their realizability has led to the result that there is no geometric $(19_4)$ configuration~\cite{BokowskiPilaud2}.
\item For larger values of~$n$, one cannot expect a complete classification of $(n_4)$ configurations. However, one can construct families of examples of $4$-configurations. One of the key ingredients for such constructions is the use of symmetries. See Figure~\ref{fig:splittings20_4} for the smallest example obtained in this way, and refer to the detailed presentation in Gr\"unbaum's recent monograph~\cite{Grunbaum1}.
\item Finally, Bokowski and Schewe introduced in~\cite{BokowskiSchewe2} a method to produce $(n_4)$ configurations from deficient configurations. It consists in finding two point\,--\,line incidence structures $(P,L)$ and~$(P',L')$ of respective signatures~$(ax^3+bx^4,cy^3+dy^4)$ and~${(cx^3+ex^4,ay^3+fy^4)}$, where $a+b+c+e = a+c+d+f = n$, and a projective transformation which sends the $3$-valent points of~$P$ to points contained in a $3$-valent line of~$L'$, and at the same time the $3$-valent lines of~$L$ to lines containing a $3$-valent point of~$P'$. This method was used to obtain the first examples of~$(29_4)$ and~$(31_4)$ configurations.
\end{enumerate}

In this paper, we are interested in this very last method described above. We are going to study deficient configurations (see the notion of quasi-configuration and $3|4$-configuration in the next subsection) for the use of them as building blocks for configurations. Our study has led in particular to first examples of~$(37_4)$ and $(43_4)$ configurations. Thus the remaining undecided cases for the existence of $(n_4)$ configurations are now only the cases $n = 22, 23$, and~$26$.


\subsection{Quasi-configurations}

A \defn{quasi-configuration}~$(P,L)$ is a point\,--\,line incidence structure in which each point is contained in at least~$3$ lines and each line contains at least~$3$ points of~$P$. In other words, the signature~$(\bP,\bL)$ of~$(P,L)$ satisfies~$x^3 \,|\, \bP(x)$ and~$y^3 \,|\, \bL(y)$.  The term \mbox{``quasi-configuration''} for this concept was suggested by Gr\"unbaum to the authors. As observed above, these quasi-configurations can sometimes be used as building blocks for larger point\,--\,line incidence structures.

In this paper, we investigate in particular \defn{$3|4$-configurations}, where each point of~$P$ is contained in~$3$ or~$4$ lines of~$L$ and each line of~$L$ contains~$3$ or~$4$ points of~$P$. In other words, $3|4$-configurations are point\,--\,line incidence structures whose signature is of the form~$(ax^3 + bx^4, cx^3 + dx^4)$ for some~$a,b,c,d \in \N$ satisfying~$3a+4b = 3c+4d$. Note that their numbers of points and lines do not necessarily coincide. If it is the case, \ie, if $a+b = c+d = n$, we speak of an \defn{$(n_{3|4})$ configuration}. In this case, $a=c$ and~$b=d$, the number of points and lines is~$n = a+b = c+d$ and the number of incidences is~$3a+4b = 3c+4d$.

We think of an $(n_{3|4})$ configuration as a deficient $(n_4)$ configuration. A good measure on~$(n_{3|4})$ configurations is the number of missing incidences~$a$. We say that an $(n_{3|4})$ configuration is \defn{optimal} if it contains the maximal number of point\,--\,line incidences among all $(n_{3|4})$ configurations. One objective is to study and classify optimal $(n_{3|4})$ configurations for small values of~$n$. 

\begin{example}
\label{exm:8282}
Figure~\ref{fig:8282} shows a point\,--\,line incidence structure with signature~${(8x^3+2x^4, 8y^3+2y^4)}$. It is a $10_{3|4}$-configuration: the $3$-valent elements are colored red while the $4$-valent elements are colored blue. We will see in Figure~\ref{fig:optimal} that this $3|4$-configuration is not optimal.

\begin{figure}[h]
  \centerline{\includegraphics[width=.7\textwidth]{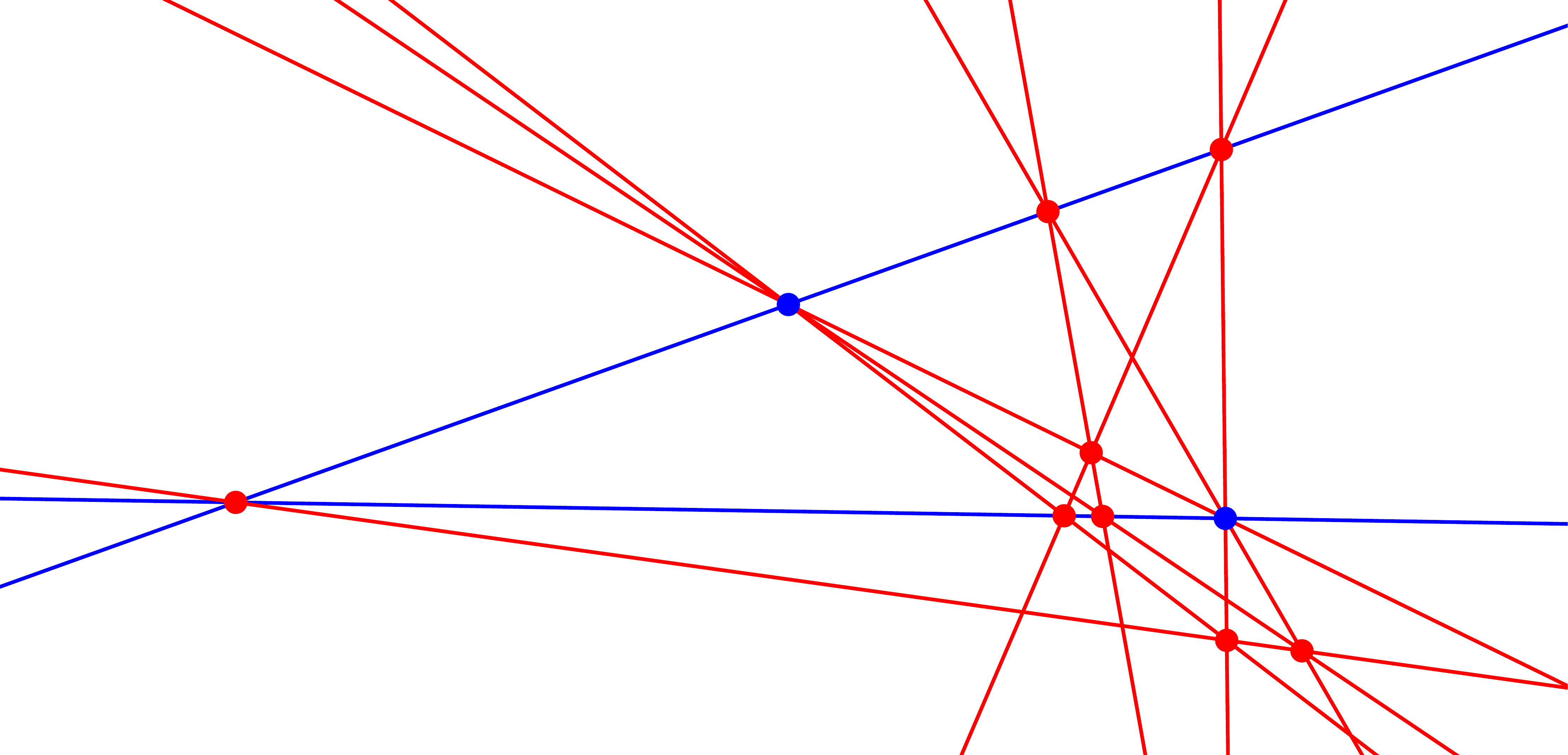}}
  \caption{A quasi-configuration with signature~$(8x^3+2x^4, 8y^3+2y^4)$.}
  \label{fig:8282}
\end{figure}
\end{example}


\subsection{Overview}

The paper is divided into two parts. In Section~\ref{sec:constructions}, we illustrate how quasi-configurations (in particular $3|4$-configurations) can be used as building blocks to construct $(n_4)$ configurations, and we obtain in particular examples of $(37_4)$ and $(43_4)$ configurations. In Section~\ref{sec:obstructions}, we present a counting obstruction for the existence of topological quasi-configurations, and we study optimal $(n_{3|4})$ configurations with few points~and~lines.


\section{Constructions}
\label{sec:constructions}

In this section, we discuss different ways to obtain new point\,--\,line incidence structures from old ones. We are in particular interested in the construction of new quasi-configurations from existing quasi-configurations. We use these techniques to provide the first $(37_4)$ and $(43_4)$ configurations.


\subsection{Operations on point\,--\,line incidence structures}

To construct new point\,--\,line incidence structures from old ones, we will use the following operations, illustrated in Section~\ref{subsec:examples}:
\begin{description}
\item[Deletion] Deleting elements from a point\,--\,line incidence structure yields a smaller incidence structure. Note that deletions do not necessarily preserve connectedness or quasi-configurations. We can however use deletions in~$4$-configurations to construct $3|4$-configu\-rations if no remaining element is incident to two deleted elements.
\item[Addition] As illustrated by the example of Gr\"unbaum's $(20_4)$ configuration~\cite{Grunbaum5} in Figure~\ref{fig:splittings20_4}, certain point\,--\,line incidence structures can be obtained as the disjoint union of two smaller incidence structures~$(P,L)$ and~$(P',L')$. In particular, we obtain an $(n_4)$ configuration if $(P,L)$ and~$(P',L')$ are $3|4$-configurations, if each $3$-valent element of~$(P,L)$ is incident to precisely one $3$-valent element of~$(P',L')$ and conversely, and if no other incidences appear.
\item[Splitting] The reverse operation of addition is splitting: given a point\,--\,line incidence structure, we can split it into two smaller incidence structures. We can require additionally the two resulting incidence structures to be quasi-configurations or even regular configurations. For example, the two geometric $(18_4)$ configurations~\cite{BokowskiSchewe2, BokowskiPilaud} as well as Gr\"unbaum's $(20_4)$ configuration~\cite{Grunbaum5} are splittable into $3|4$-configurations, see Figures~\ref{fig:splittings20_4} and~\ref{fig:splittings18_4}.
\item[Superposition] Slightly more general than addition is the superposition, where we allow the two point\,--\,line incidence structures~$(P,L)$ and~$(P',L')$ to share points or lines. For example, we can superpose two $2$-valent vertices to make one $4$-valent vertex. This idea is used in our construction of $(37_4)$ and $(43_4)$ configurations below.
\end{description}

\begin{figure}[h]
  \centerline{\includegraphics[width=.5\textwidth]{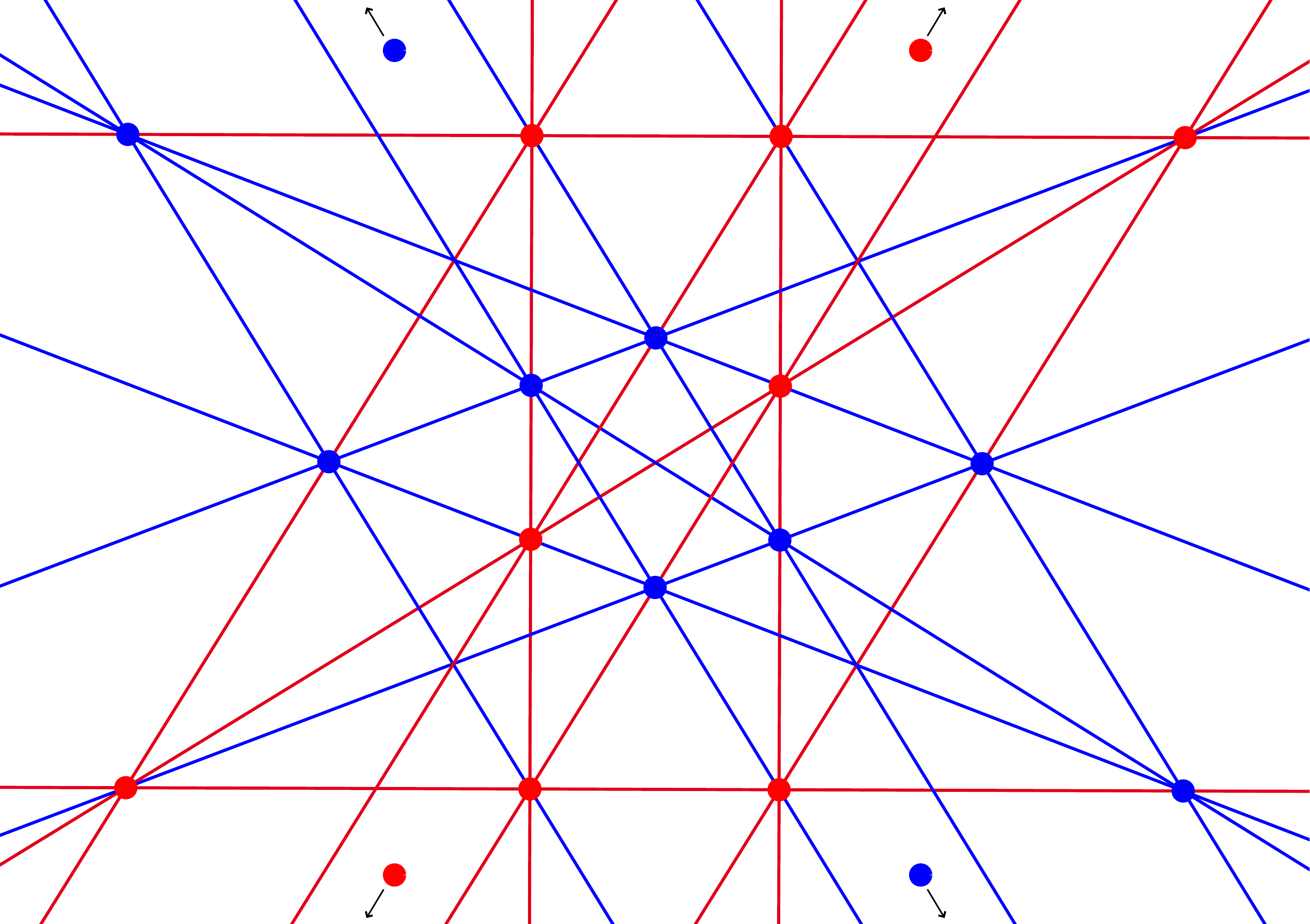}\hspace{.2cm}\includegraphics[width=.5\textwidth]{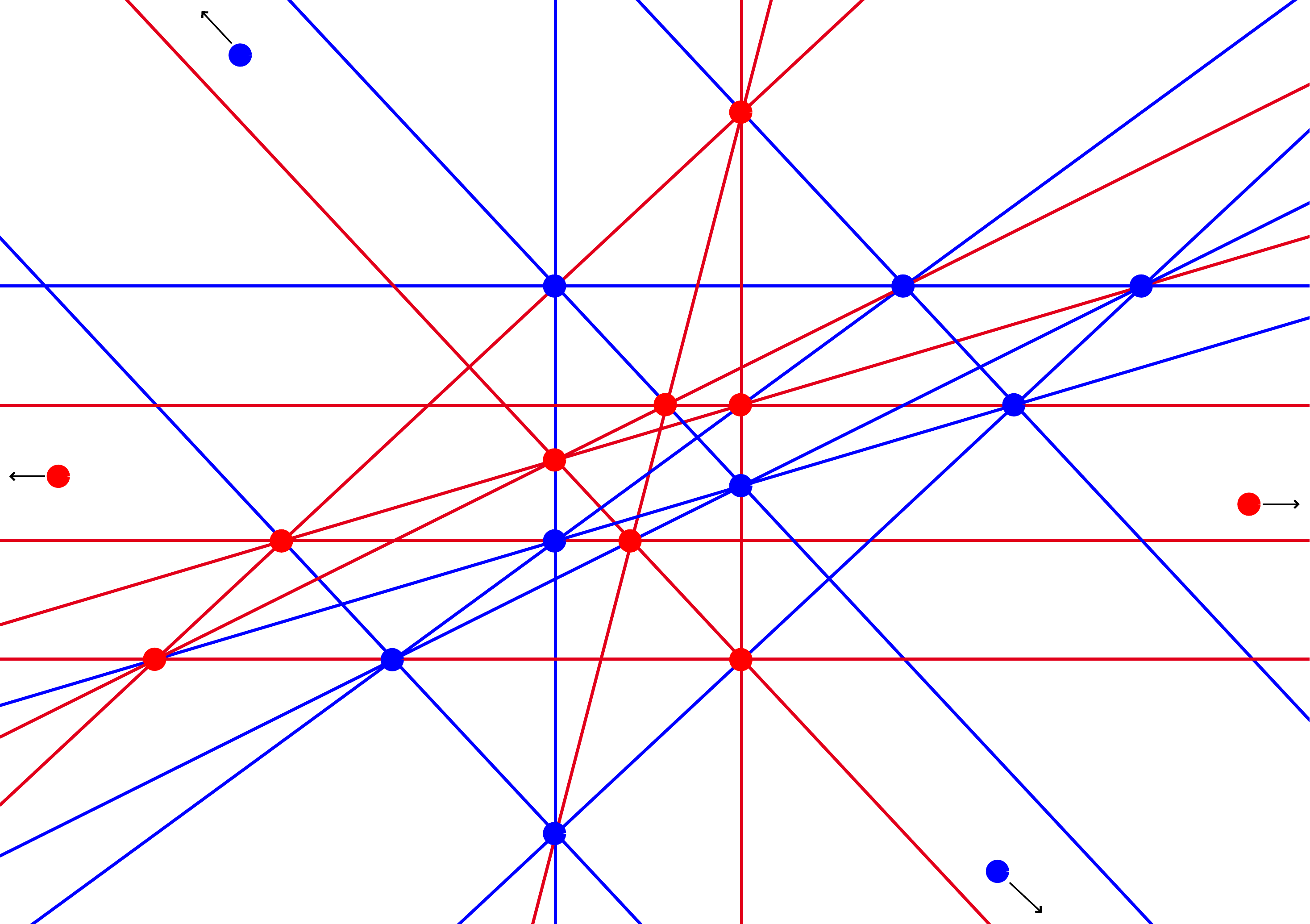}}
  \caption{Splittings of the two geometric $(18_4)$ configurations~\cite{BokowskiSchewe2, BokowskiPilaud} into two $(9_{3|4})$ configurations. The rightmost $(18_4)$ configuration is even split into two $(9_3)$ configurations. In both pictures, the points which seem isolated are in fact at infinity in the direction pointed by the corresponding arrow, and are incident to the $4$ lines parallel to that direction.}
  \label{fig:splittings18_4}
\end{figure}


\subsection{Examples of constructions}
\label{subsec:examples}

We now illustrate the previous operations and produce $4$-configurations from smaller point\,--\,line incidence structures. We start with a simple example.

\begin{example}[A $(38_4)$ configuration]
\label{exm:38_4}
It was shown in~\cite{BokowskiPilaud2} that no topological $(19_4)$ configuration can be geometrically realized with points and lines in the projective plane. However, Figure~\ref{fig:38_4}\,(left) shows a geometric realization of a topological $(19_4)$ configuration where one line has been replaced by a circle. Forgetting this circle, we obtain a $3|4$-configuration with signature~$(15x^4+4x^3,18x^4)$. We take two opposite copies of this $3|4$-configuration (colored purple and red in Figure~\ref{fig:38_4}\,(right)) and add two lines (colored green in Figure~\ref{fig:38_4}\,(right)) each incident to two points in each copy. We obtain a $(38_4)$ configuration.

\begin{figure}[h]
  \centerline{\includegraphics[width=\textwidth]{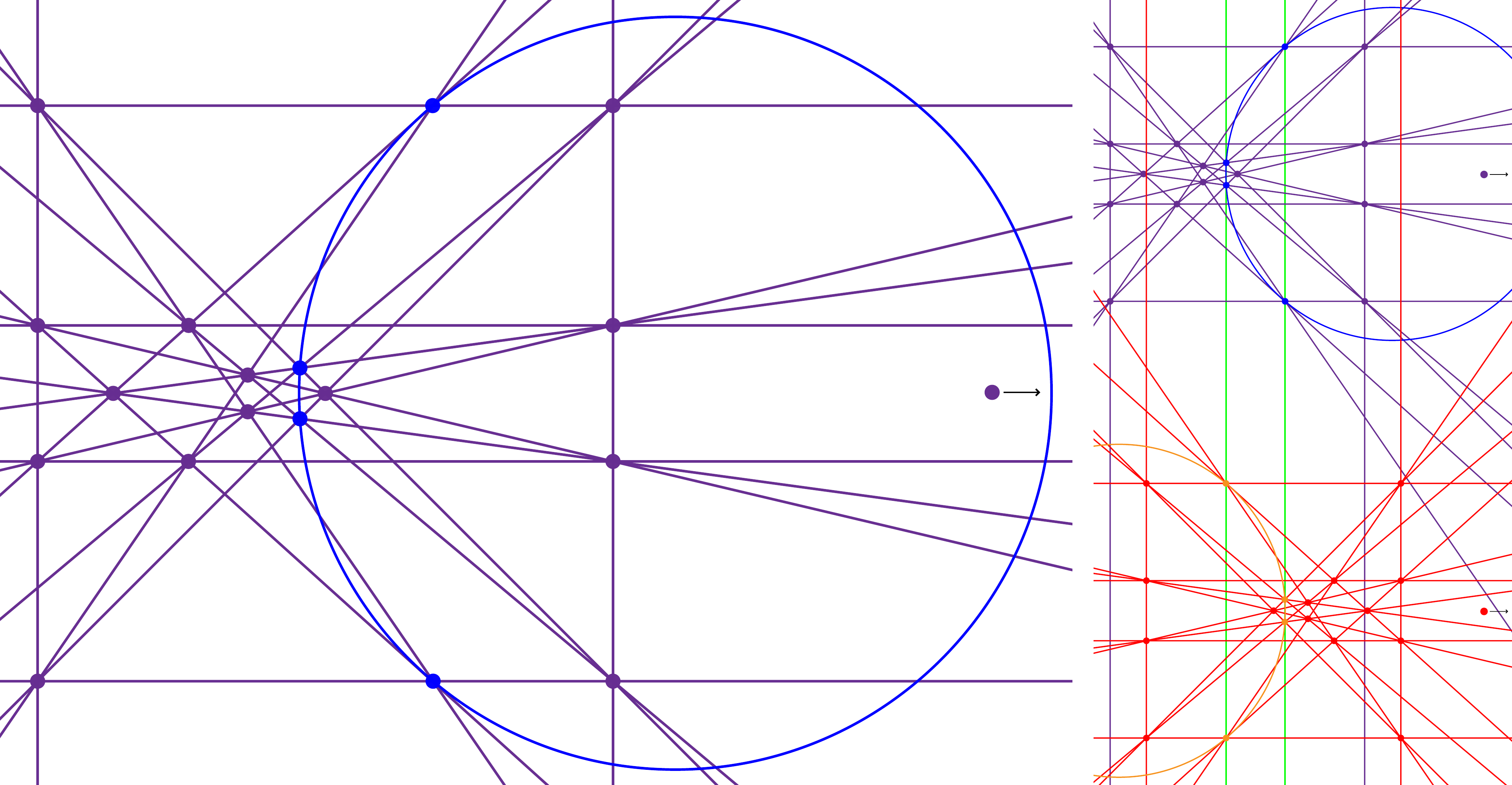}}
  \caption{(Left) A geometric realization of a topological $(19_4)$ configuration where one line has been replaced by a circle. (Right) A $(38_4)$ configuration built from two copies of this incidence structure. The construction is explained in full detail in Example~\ref{exm:38_4}.}
  \label{fig:38_4}
\end{figure}
\end{example}

Using similar ideas, we now observe that it is always possible to produce a $4$-configuration from any $3|4$-configuration.

\begin{example}[Any $3|4$-configuration generates a $4$-configuration]
From a $3|4$-configuration with signature~$(ax^3 + bx^4, cy^3 + dy^4)$, we construct an $(n_4)$ configuration with $n = 16a + 16b + 4c = 4a + 16c + 16d$ as follows:
\begin{enumerate}[(i)]
\item We take four translated copies of the $3|4$-configuration and add suitable parallel lines through all $3$-valent points.
\item We take the geometric dual of the resulting $3|4$-configuration (remember that geometric duality transforms a point~$p$ of the projective plane into the line formed by all points orthogonal to~$p$ and conversely).
\item We take again four translated copies of this dual $3|4$-configuration and add suitable parallel lines through all $3$-valent vertices.
\end{enumerate}
\end{example}

Of course, we can try to obtain other $4$-configurations from $3|4$-configurations. This approach was used by Bokowski and Schewe~\cite{BokowskiSchewe2} to construct $(29_4)$ and $(31_4)$ configurations from the $(14_{3|4})$, $(15_{3|4})$ and $(16_{3|4})$ configurations of Figure~\ref{fig:seemOptimal}. We refer to their paper~\cite{BokowskiSchewe2} for an explanation. Here, we elaborate on the same idea to construct two new relevant $(n_4)$ configurations.

\begin{example}[First $(43_4)$ configuration]
\label{exm:43_4}
To construct a $((n+m)_4)$ configuration from an $(n_4)$ configuration and an $(m_4)$ configuration, we proceed as follows (see Figure~\ref{fig:43_4}):
\begin{enumerate}[(i)]
\item We delete two points not connected by a line in the $(n_4)$ configuration and consider the eight resulting $3$-valent lines (colored blue in Figure~\ref{fig:43_4}\,(top left) and orange in Figure~\ref{fig:43_4}\,(top right)).
\item We add four points (colored green in Figure~\ref{fig:43_4}), each incident with precisely two $3$-valent lines. All points and lines are now $4$-valent again, except the four new $2$-valent points.
\item We do the same operations in the $(m_4)$ configuration.
\item Finally, we use a projective transformation that maps the set of four $2$-valent points in the first quasi-configuration onto the set of four $2$-valent points in the second quasi-configuration. This transformation superposes the $2$-valent points to make them $4$-valent.
\end{enumerate}
If this transformation does not superpose other elements than the $2$-valent ones and does not create additional unwanted incidences, it yields the desired $((n+m)_4)$ configuration. This construction is illustrated on Figure~\ref{fig:43_4}, where we obtain a $(43_4)$ configuration from a $(25_4)$ configuration~\cite{Grunbaum1} and a $(18_4)$ configuration~\cite{BokowskiPilaud}.

\begin{figure}
  \centerline{\includegraphics[scale=.5]{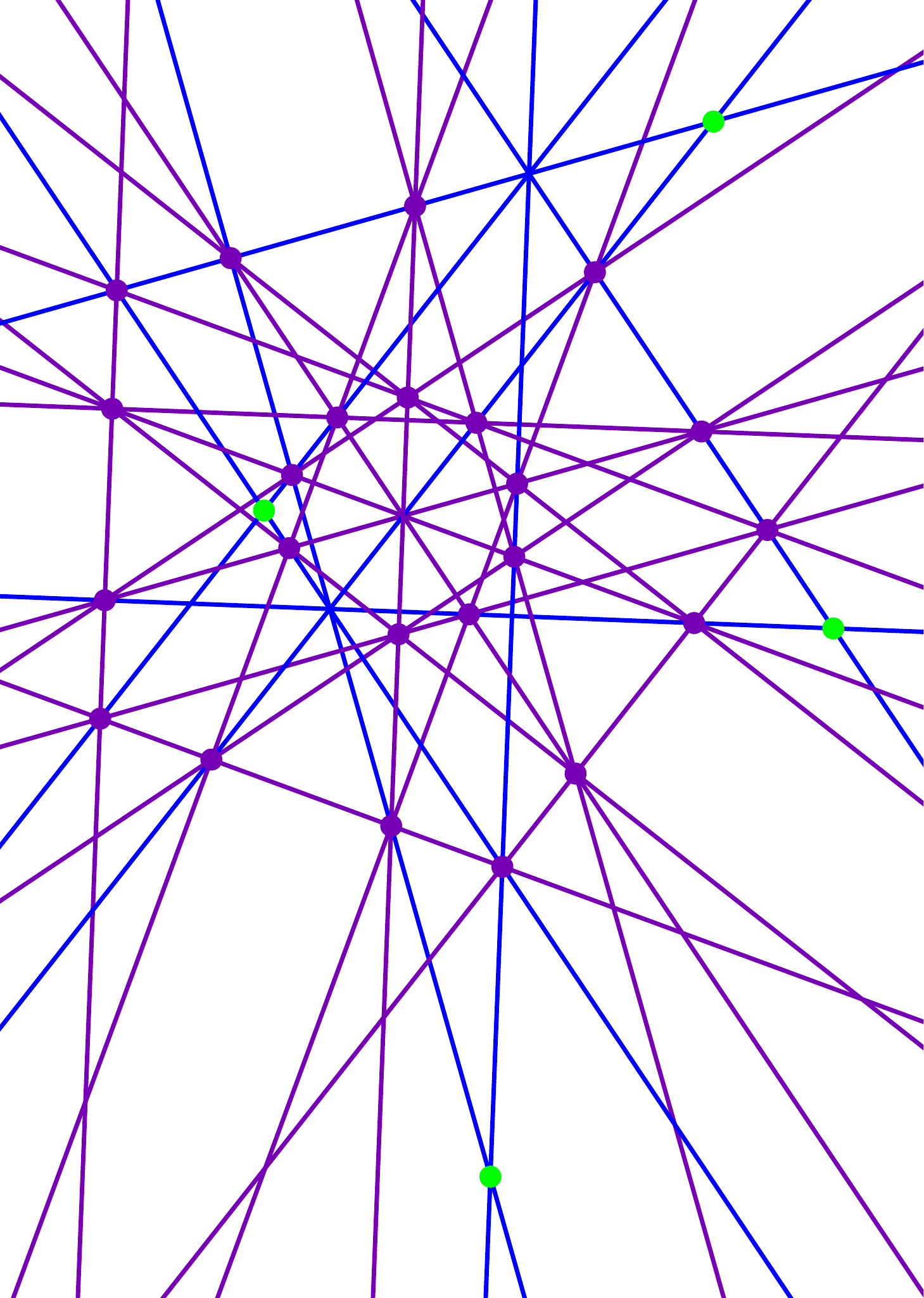}\hspace{.2cm}\includegraphics[scale=.5]{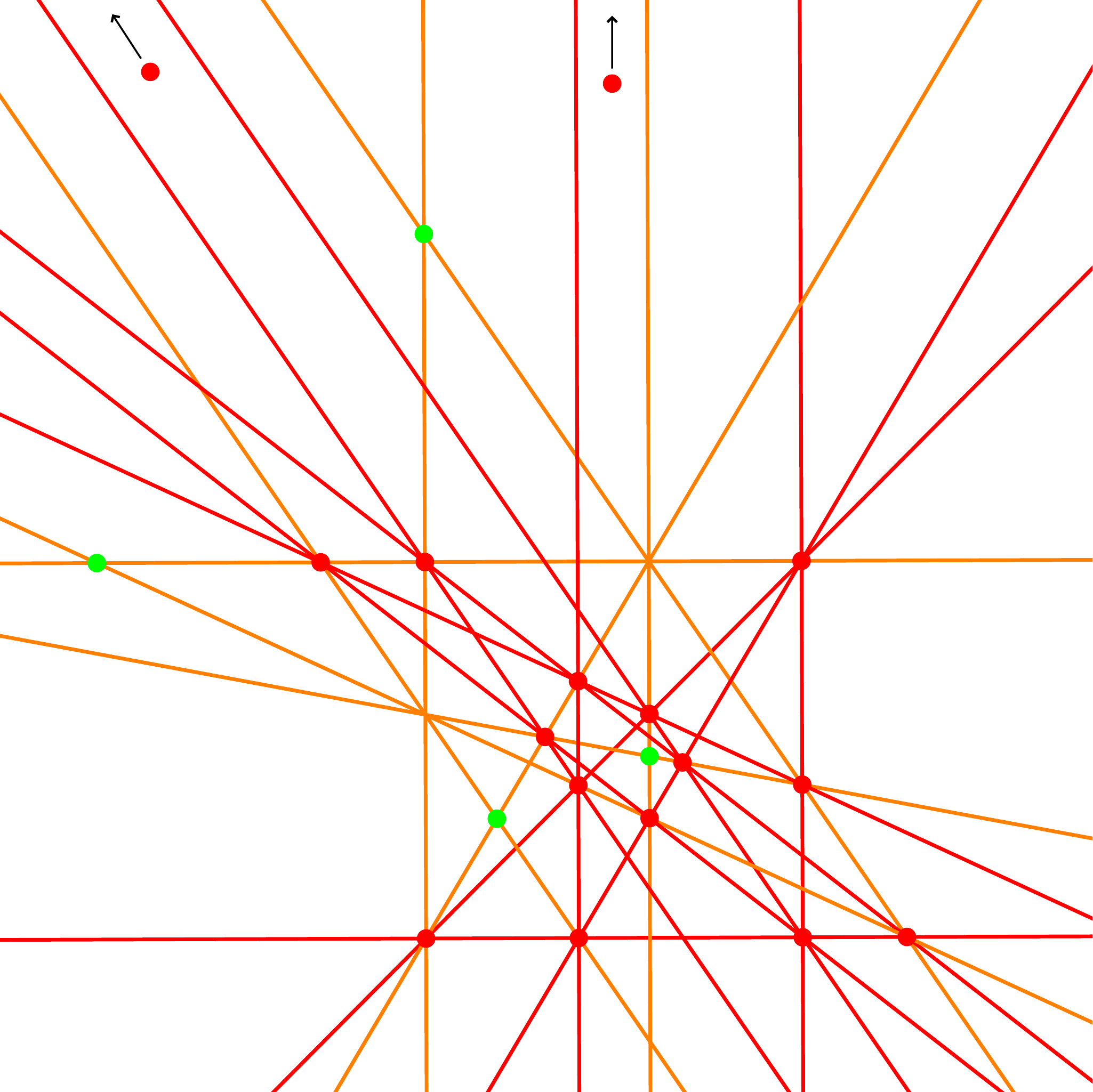}}
  \vspace{.2cm}
  \centerline{\includegraphics[scale=.5]{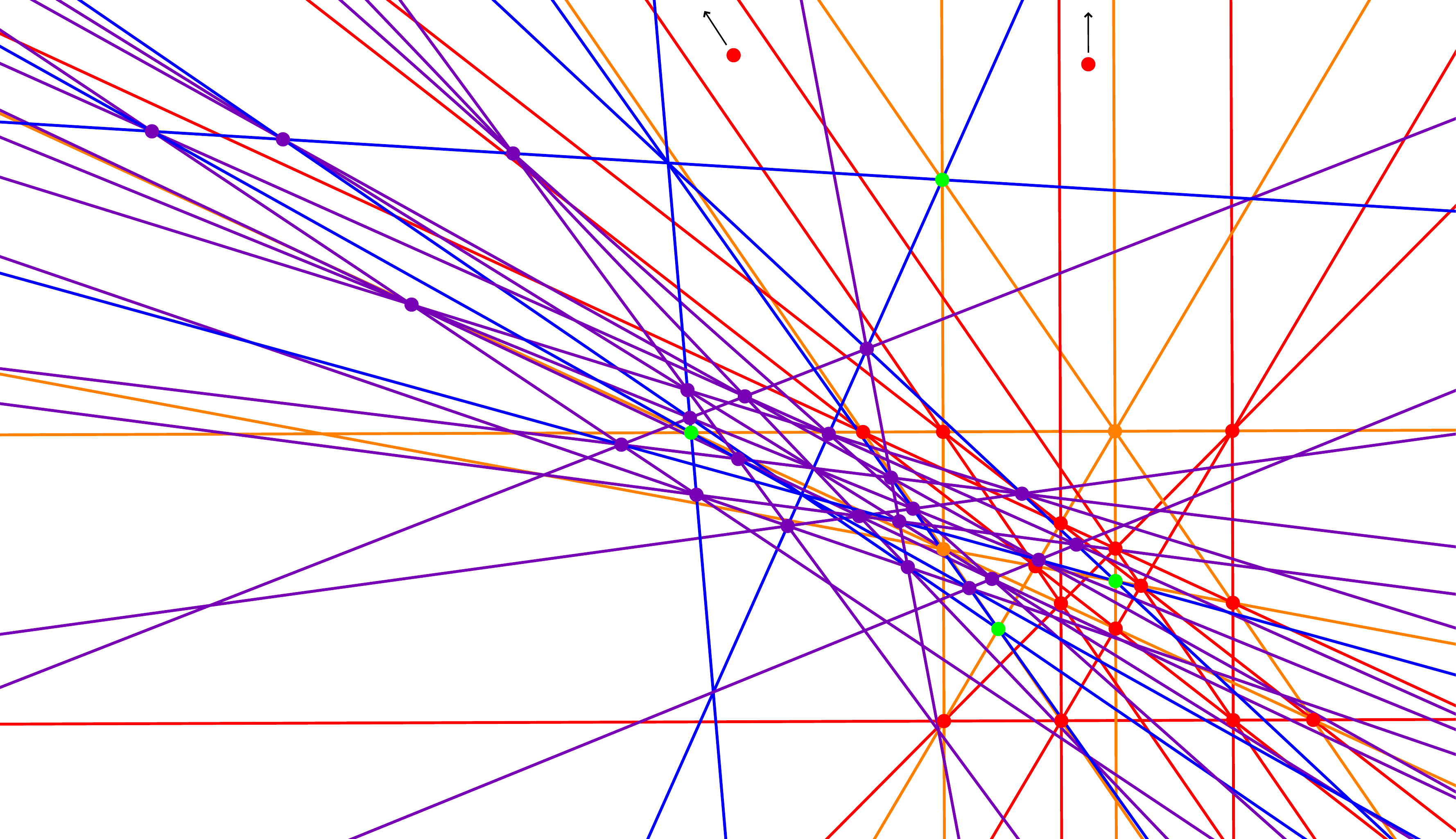}}
  \caption{A $(43_4)$ configuration built from deficient $(25_4)$ and $(18_4)$ configurations. The construction is explained in full detail in Example~\ref{exm:43_4}.}
  \label{fig:43_4}
\end{figure}
\end{example}

Unfortunately, the method from the previous example cannot provide a $(37_4)$ configuration since there is no $(n_4)$ configuration for~$n \le 17$~\cite{BokowskiSchewe2} and for~$n = 19$~\cite{BokowskiPilaud2}. We therefore need another method, which we describe in the following example.

\begin{example}[First $(37_4)$ configuration]
\label{exm:37_4}
To construct a $((n+m-1)_4)$ configuration from an $(n_4)$ configuration and an $(m_4)$ configuration, we proceed as follows (see Figure~\ref{fig:37_4}):
\begin{enumerate}[(i)]
\item We delete two points on the same line (colored green in Figure~\ref{fig:37_4}) of the $(n_4)$ configuration and consider the six resulting $3$-valent lines (colored blue in Figure~\ref{fig:37_4}\,(top left) and orange in Figure~\ref{fig:37_4}\,(top right)).
\item We add three points (colored green in Figure~\ref{fig:37_4}), each incident with precisely two $3$-valent lines. All points and lines are now $4$-valent again, except the initial $2$-valent line and the three new $2$-valent points.
\item We do the same operations in the $(m_4)$ configuration.
\item Finally, we use a projective transformation that maps the set of four $2$-valent elements in the first quasi-configuration onto the set of four $2$-valent elements in the second quasi-configuration. This transformation superposes the $2$-valent elements to make them $4$-valent.
\end{enumerate}
If this transformation does not superpose other elements than the $2$-valent ones and does not create additional unwanted incidences, it yields the desired $((n+m-1)_4)$ configuration. This construction is illustrated on Figure~\ref{fig:37_4}, where we obtain a $(37_4)$ configuration from a $(20_4)$ configuration~\cite{Grunbaum1} and a $(18_4)$ configuration~\cite{BokowskiPilaud}.

\begin{figure}
  \centerline{\includegraphics[scale=.42]{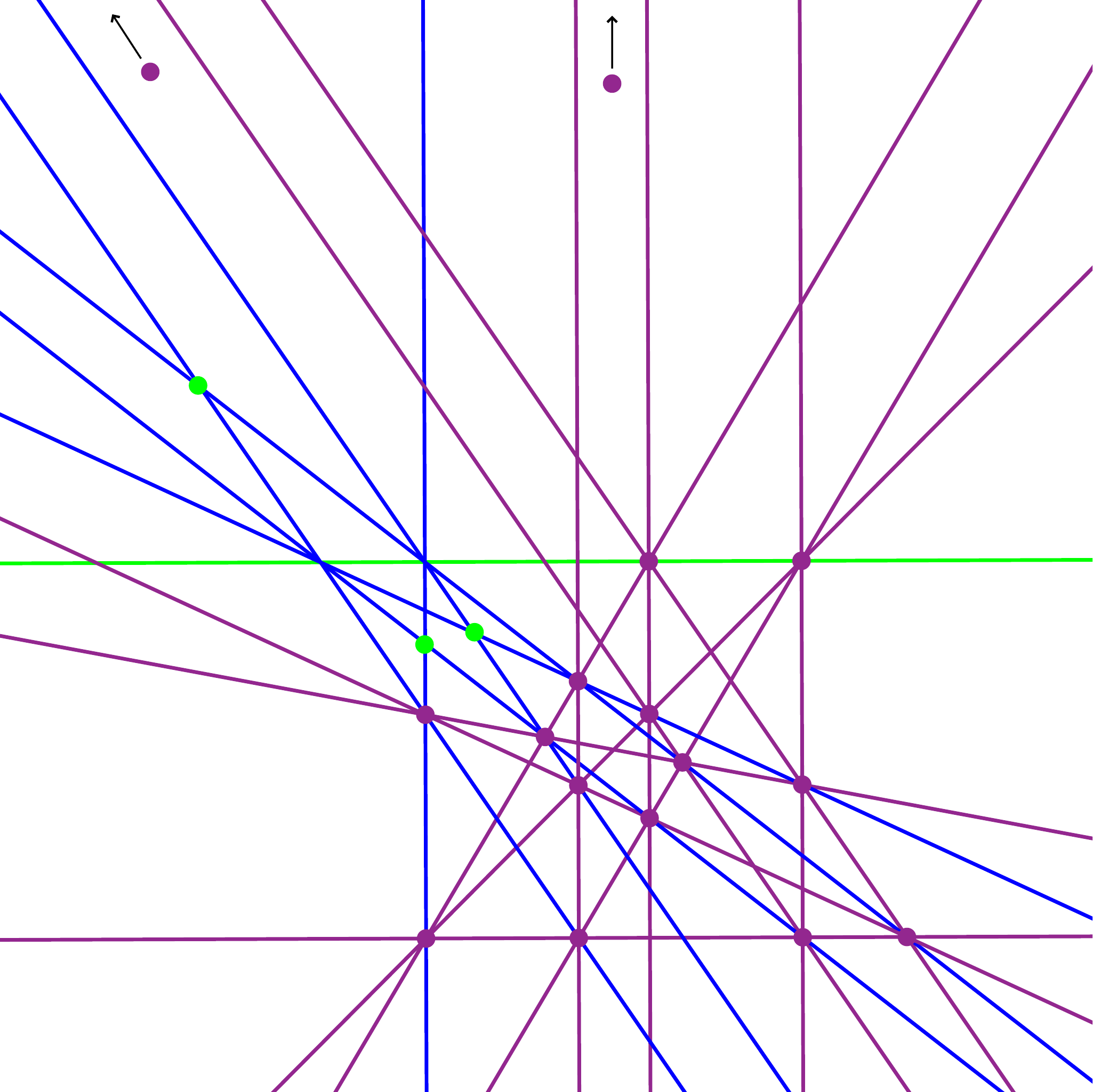}\hspace{.2cm}\includegraphics[scale=.365]{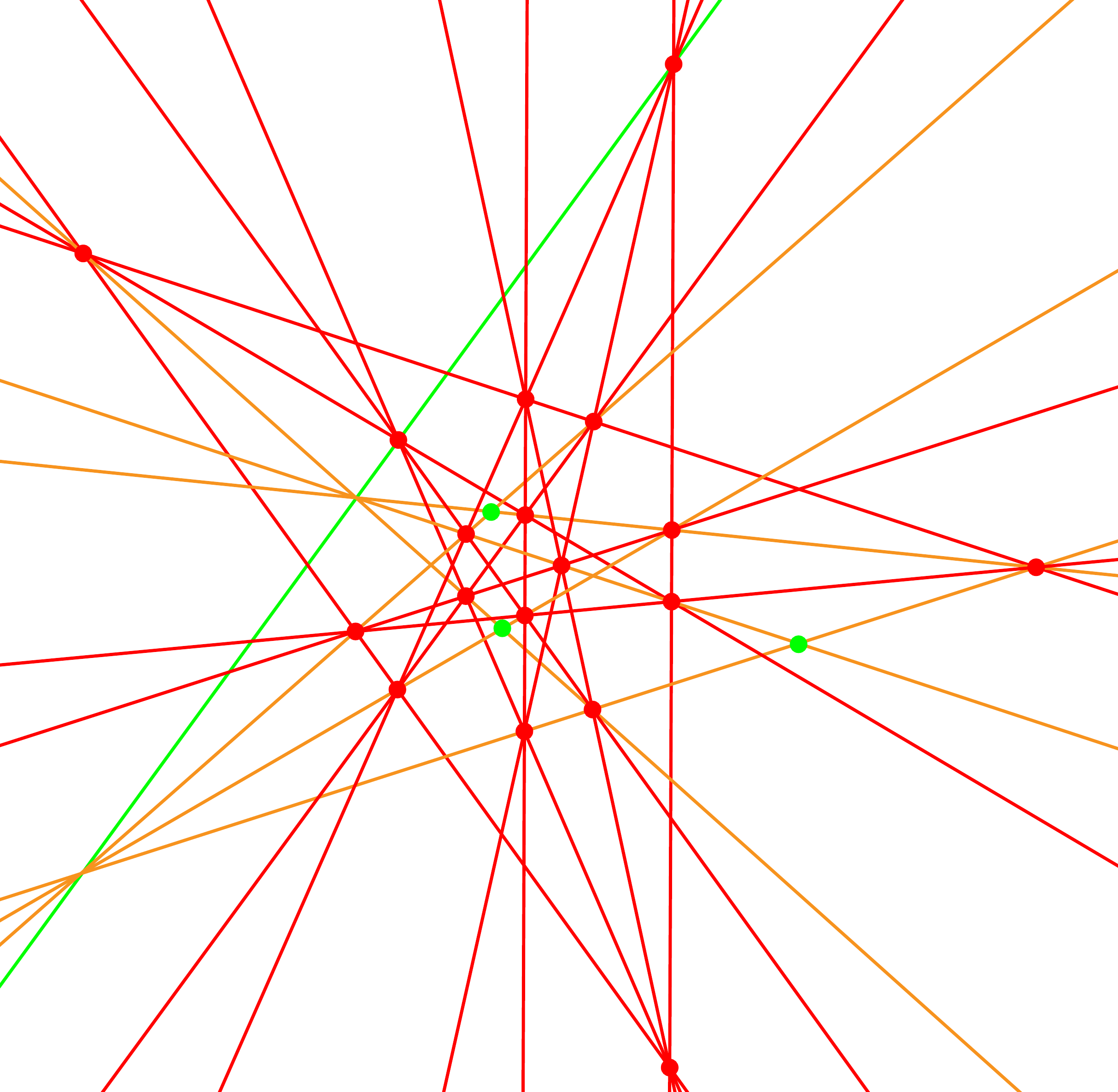}}
  \vspace{.2cm}
  \centerline{\includegraphics[scale=.43]{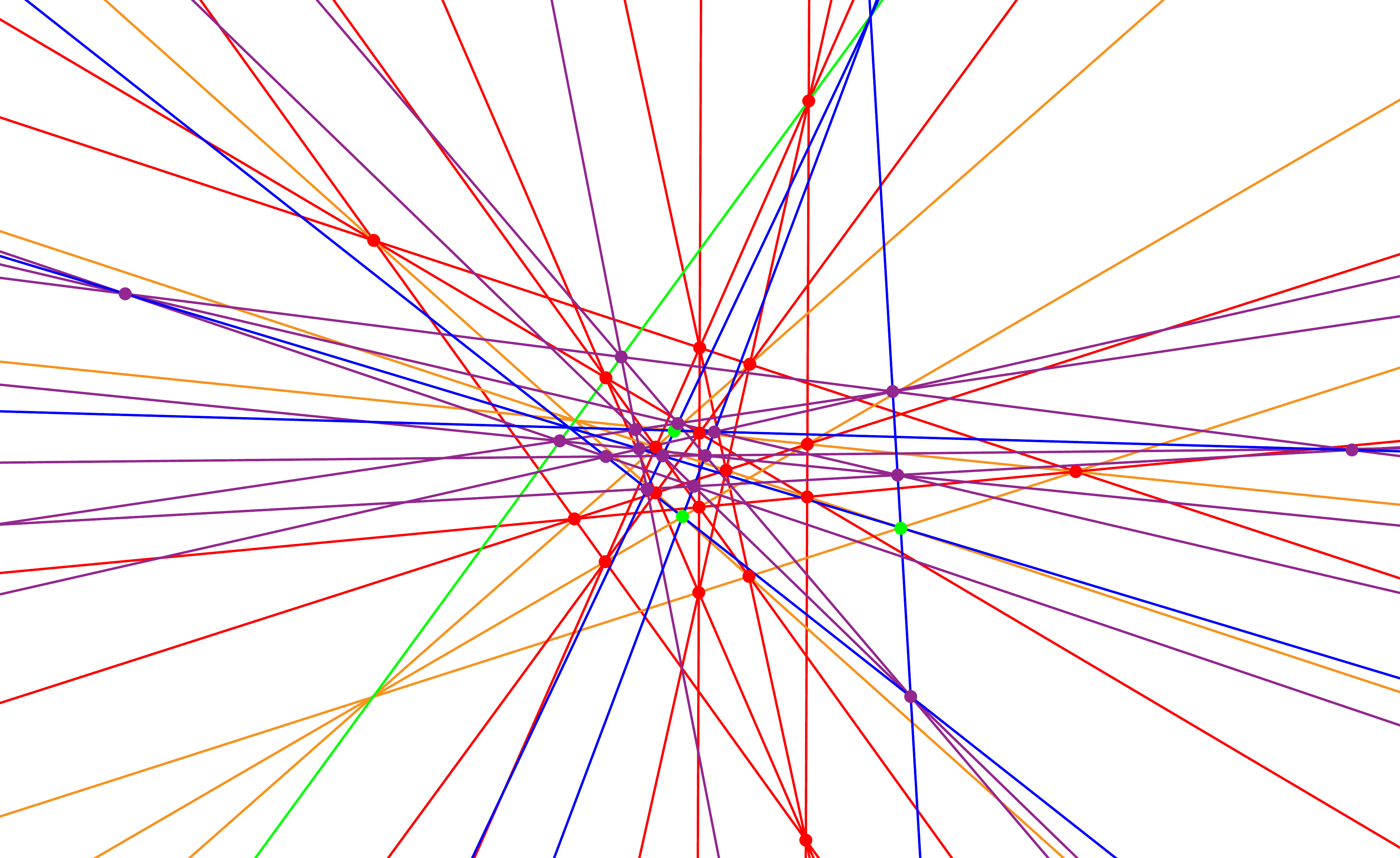}}
  \caption{A $(37_4)$ configuration built from deficient $(20_4)$ and $(18_4)$ configurations. The construction is explained in full detail in Example~\ref{exm:37_4}.}
  \label{fig:37_4}
\end{figure}
\end{example}

We invite the reader to try his own constructions, similar to the constructions of Examples~\ref{exm:43_4} and~\ref{exm:37_4}, using the operations on point\,--\,line incidence structures described above. In this way, one can obtain many $(n_4)$ configurations for various values of~$n$. Additional features can even be imposed, such as non-trivial motions or symmetries. We have however not been able to find answers to the following question.

\begin{question}
Can we create a $(22_4)$ configuration by glueing two quasi-configurations with $11$ points and lines each? More generally, can we construct $(22_4)$, $(23_4)$, or $(26_4)$ configurations by superposition of smaller quasi-configurations?
\end{question}


\section{Obstructions and optimal $3|4$-configurations}
\label{sec:obstructions}

In this section, we further investigate point\,--\,line incidence structures and $3|4$-configurations. We start with a necessary condition for the existence of topological incidence structures with a given signature. For this, we extend to all topological incidence structures an argument of Bokowski and Schewe~\cite{BokowskiSchewe1} that was used to prove the non-existence of~$(15_4)$ configurations. We obtain the following inequality.

\begin{proposition}
\label{prop:obstruction}
Let~$(P,L)$ be topological incidence structure with signature~$(\bP,\bL)$. Then
\[
\bP''(1) + 2\bP'(1) - \bL(1)^2 + \bL(1) - 6\bP(1) + 6 \le 0.
\]
\end{proposition}

\begin{proof}
Let~$p_i$ denote the number of $i$-valent points and~$\ell_j$ the number of $j$-valent lines in the incidence structure~$(P,L)$. The signature~$(\bP,\bL)$ is given by~$\bP(x) \eqdef \sum_i p_i x^i$ and~$\bL(y) \eqdef \sum_j \ell_j y^j$.

Since the incidence structure is topological, we can draw it on the projective plane such that no three pseudolines pass through a point which is not in~$P$. We call \defn{additional $2$-crossings} the intersection points of two lines of~$L$ which are not points of~$P$. We consider the lifting of this drawing on the $2$-sphere. We obtain a graph embedded on the sphere, whose vertices are all points of~$P$ together with all additional $2$-crossings, whose edges are the segments of lines of~$L$ located between two vertices, and whose faces are the connected components of the complement of~$L$. Let~$f_0$, $f_1$ and~$f_2$ denote respectively the number of vertices, edges and faces of this map. Denoting by~$\deg(p)$ the number of lines of~$L$ containing a point~$p \in P$ and similarly by~$\deg(\ell)$ the numbers of points of~$P$ contained in a line~$\ell \in L$, we have
\begin{align*}
f_0 & = 2{\bL(1) \choose 2} - 2\sum_{p \in P} \left( \!\! {\deg(p) \choose 2} - 1 \right) = \bL(1) \big( \bL(1)-1 \big) + 2\bP(1) - \sum_i i(i-1)p_i, \\
f_1 & = 2\sum_{\ell \in L} \deg(\ell) + 2f_0 - 2\bP(1) = 2\sum_j j\ell_j + 2f_0 - 2\bP(1) = 2\sum_i ip_i + 2f_0 - 2\bP(1), \\
f_2 & = f_1 - f_0 + 2.
\end{align*}
Moreover, since no face is a digon, we have~$3f_2 \le 2f_1$. Replacing~$f_2$ and~$f_1$ by the above expressions, we obtain
\[
0 \ge 3f_2 - 2f_1 = f_1 - 3f_0 + 6 = 2\sum_i ip_i - 4\bP(1) - f_0 + 6 = \sum_i i(i+1)p_i - \bL(1) \big( \bL(1)-1 \big) - 6 \big( \bP(1)-1 \big),
\]
and thus the desired inequality.
\end{proof}

\begin{corollary}
\label{coro:ab}
If there is a topological incidence structure with signature~$(ax^3+bx^4, ay^3+by^4)$,~then
\[
-(a+b)^2+7a+15b+6 \le 0.
\]
The following table provides the minimum value of~$b$ for which there could exist a topological incidence structure with signature~$(ax^3+bx^4, ay^3+by^4)$:

\medskip
\centerline{
	\begin{tabular}{c|cccccccccc}
	$a$        & $0$  & $1$  & $2$  & $3$  & $4$ & $5$ & $6$ & $7$ \\
	\hline
	$b_{\min}$ & $16$ & $14$ & $13$ & $11$ & $9$ & $8$ & $6$ & $3$
	\end{tabular}
}
\end{corollary}

\begin{proof}
Direct application of Proposition~\ref{prop:obstruction} with~$\bP(x) = ax^3+bx^4$ and~$\bL(y) = ay^3+by^4$. 
\end{proof}

For example, there is no topological $(15_4)$ configuration~\cite{BokowskiSchewe1} and no incidence structure with signature~$(7x^3+2x^4, 7y^3+2y^4)$. Compare to Example~\ref{exm:8282} which shows that a configuration with signature~$(8x^3+2x^4, 8y^3+2y^4)$ exists.

\begin{corollary}
\label{coro:optimal}
A $(n_{3|4})$ configuration has at most~$I_{\max} \eqdef \min\left(4n \; , \; \left\lfloor\dfrac{n^2+17n-6}{8}\right\rfloor\right)$ incidences.

\vspace{-.15cm}
\noindent
The values of~$I_{\max}$ appear in the following table:

\medskip
\centerline{
	\begin{tabular}{c|cccccccccc}
		$n$ & $7$ & $8$ & $9$ & $10$ & $11$ & $12$ & $13$ & $14$ & $15$ & $16$ \\
		\hline
		$I_{\max}$ & $20$ & $24$ & $28$ & $33$ & $37$ & $42$ & $48$ & $53$ & $59$ & $64$
	\end{tabular}
}
\end{corollary}

\begin{proof}
Consider an $(n_{3|4})$ configuration with signature~$(ax^3+bx^4, ay^3+by^4)$ where~$a+b=n$. The number of incidences is~$I \eqdef 3a+4b$. It can clearly not exceed~$4n$. For the second term in the minimum, we apply Corollary~\ref{coro:ab} to get
\[
0 \ge -(a+b)^2 + 7a + 15b + 6 = - (a+b)^2 + 8(3a+4b) - 17(a+b) + 6 = -n^2+8I-17n+6. \qedhere
\]
\end{proof}

\begin{corollary}
There is no topological $(n_{3|4})$ configuration if~$n \le 8$.
\end{corollary}

\begin{proof}
If~$n \le 7$, there is no topological $(n_{3|4})$ configuration since it should have at least $3n$ incidences, which is larger than the upper bound of Corollary~\ref{coro:optimal}. If~$n = 8$, a $(8_{3|4})$ configuration should be a $(8_3)$ configuration by Corollary~\ref{coro:optimal}. But the only combinatorial $(8_3)$ configuration is not topological.
\end{proof}

To close this section, we exhibit optimal $(n_{3|4})$ configurations for small values of~$n$, \ie, $(n_{3|4})$ configurations which maximize the number of point\,--\,line incidences.

\begin{proposition}
For~$9 \le n \le 13$, the bound of Corollary~\ref{coro:optimal} is tight: there exists $(n_{3|4})$ configurations with $\left\lfloor\frac{n^2+17n-6}{8}\right\rfloor$ incidences.
\end{proposition}

\begin{proof}
For~$n=13$, we consider the $(13_{3|4})$-configuration of Figure~\ref{fig:optimal}. The homogeneous coordinates of its points and lines are given by
\[
P \eqdef L \eqdef \set{\begin{bmatrix} i \\ j \\ 1 \end{bmatrix}}{i,j \in \{-1,0,1\}} \cup \left\{\begin{bmatrix} 1 \\ 0 \\ 0 \end{bmatrix}, \begin{bmatrix} 0 \\ 1 \\ 0 \end{bmatrix}, \begin{bmatrix} 1 \\ 1 \\ 0 \end{bmatrix}, \begin{bmatrix} 1 \\ -1 \\ 0 \end{bmatrix}\right\}.
\]

\begin{figure}
  \centerline{
  \includegraphics[width=.35\textwidth]{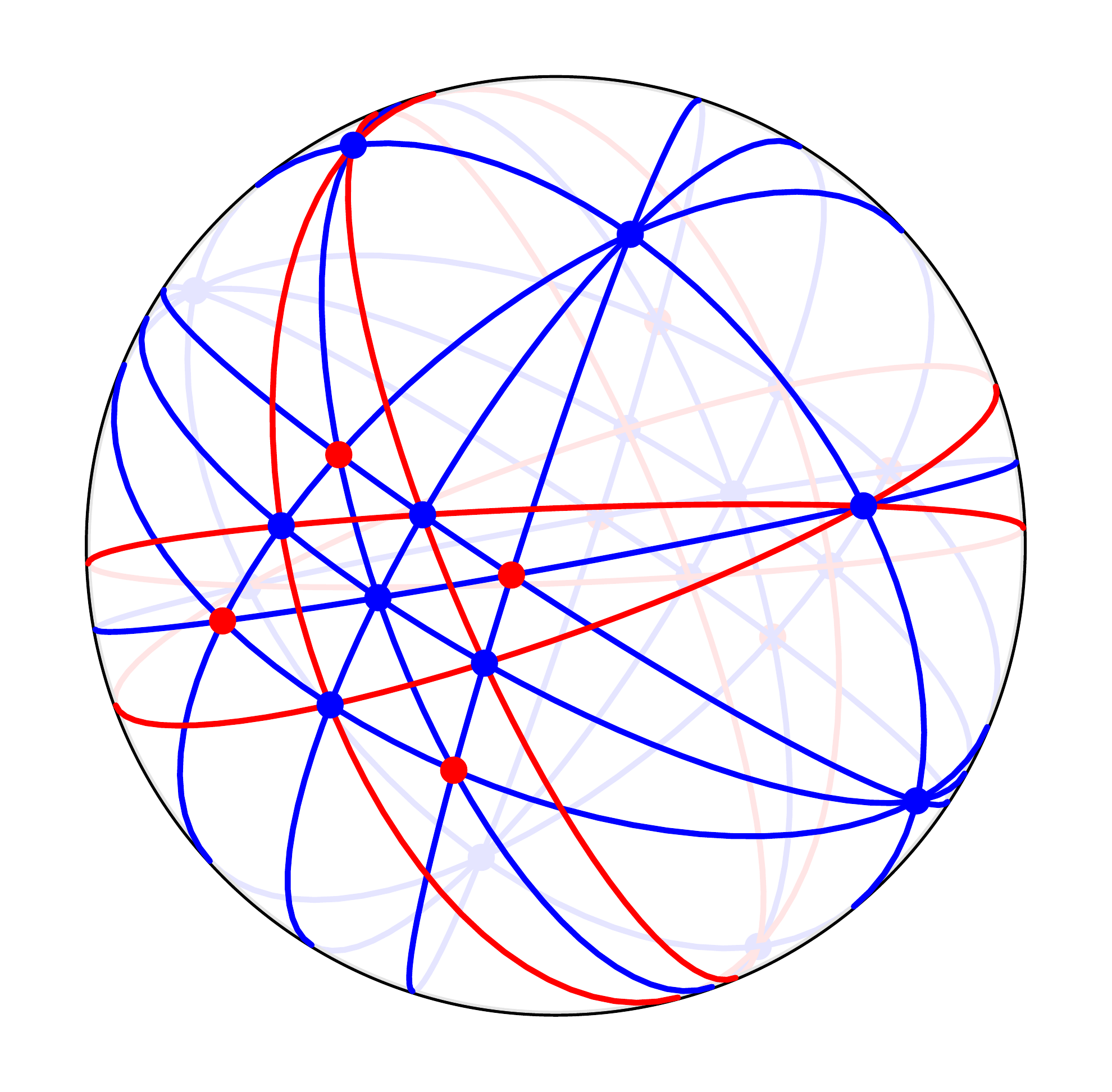}\hspace{-.5cm}
  \includegraphics[width=.35\textwidth]{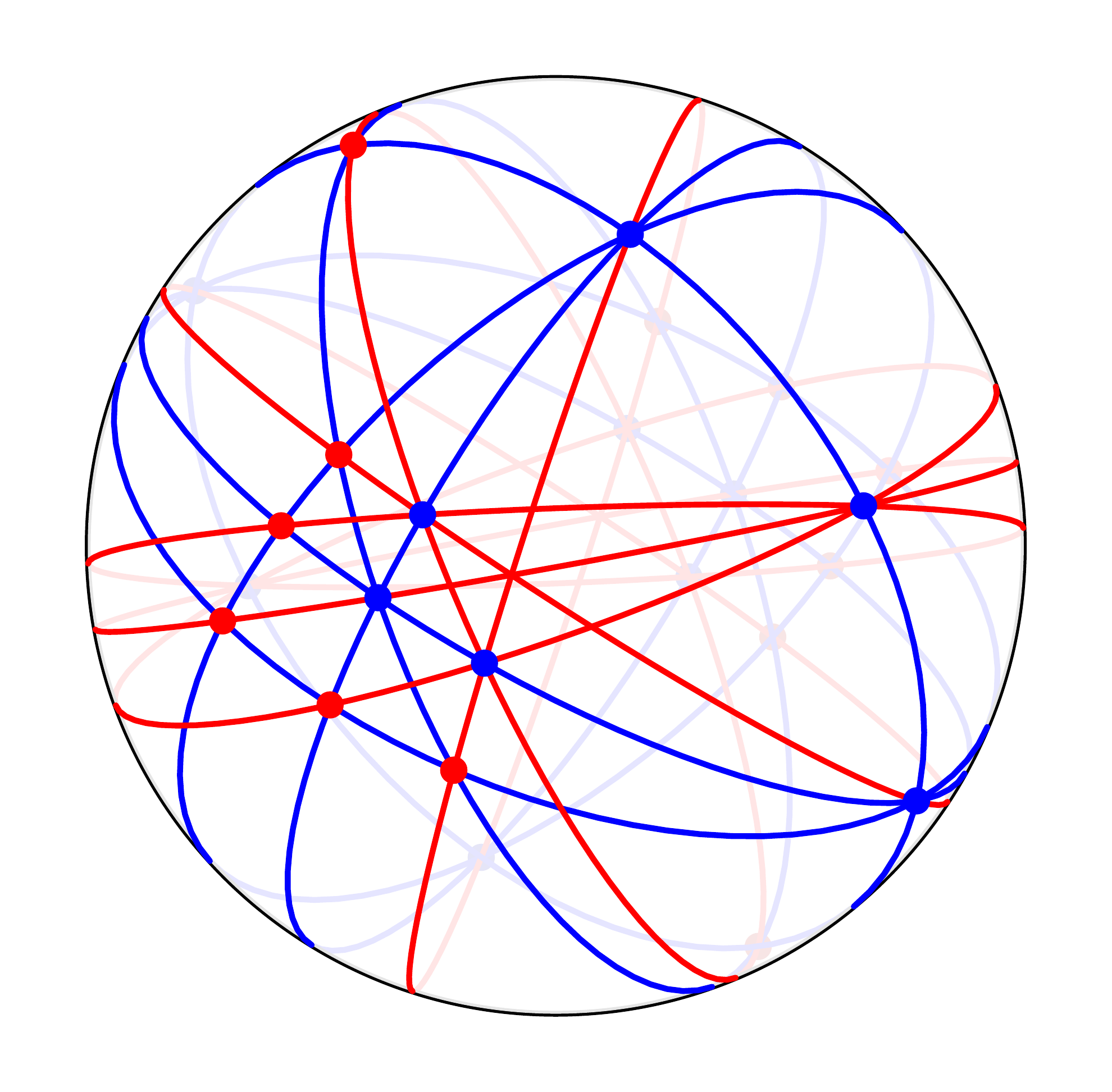}\hspace{-.5cm}
  \includegraphics[width=.35\textwidth]{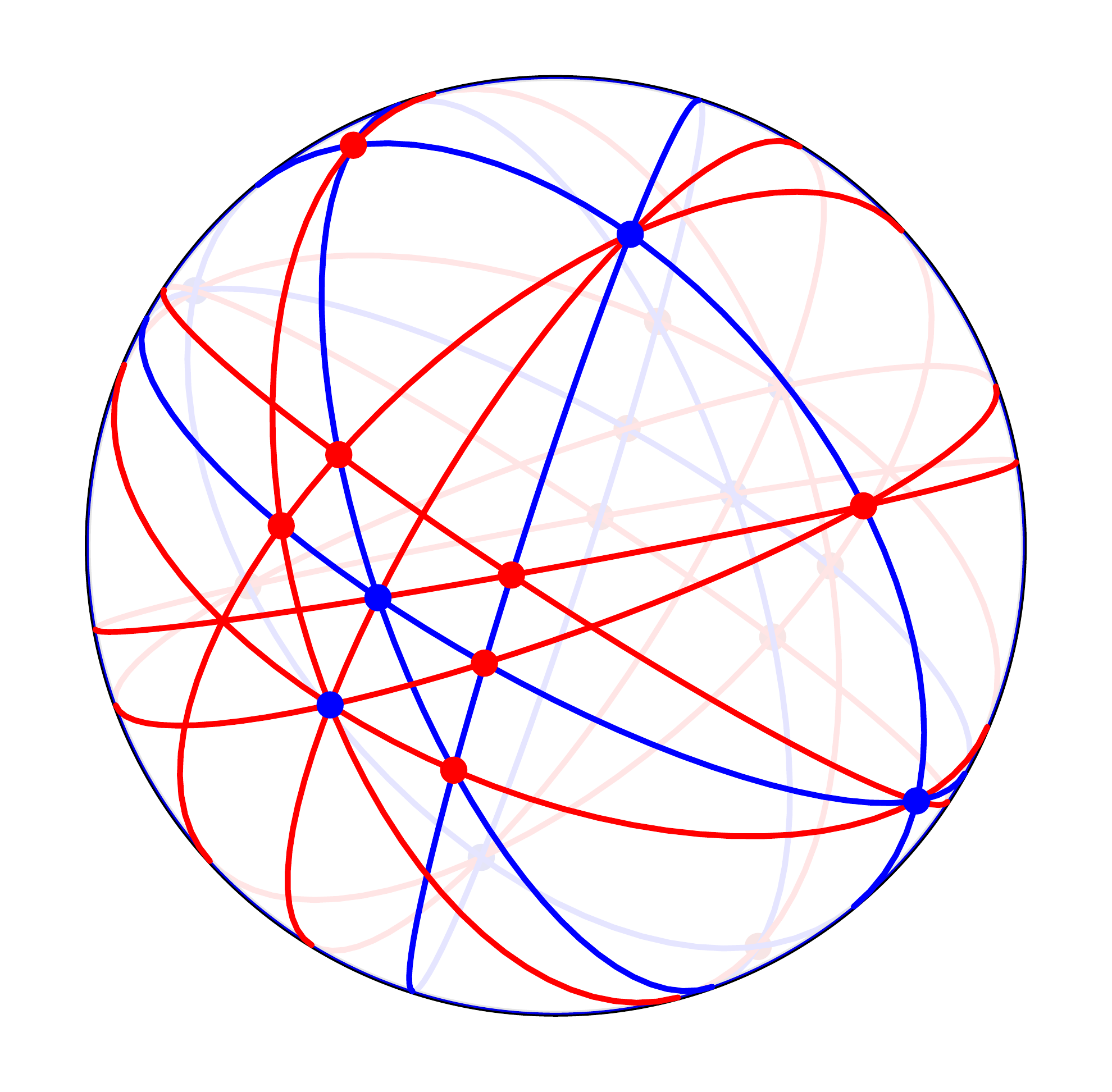}
  }
  \vspace{-1cm}
  \centerline{
  \includegraphics[width=.35\textwidth]{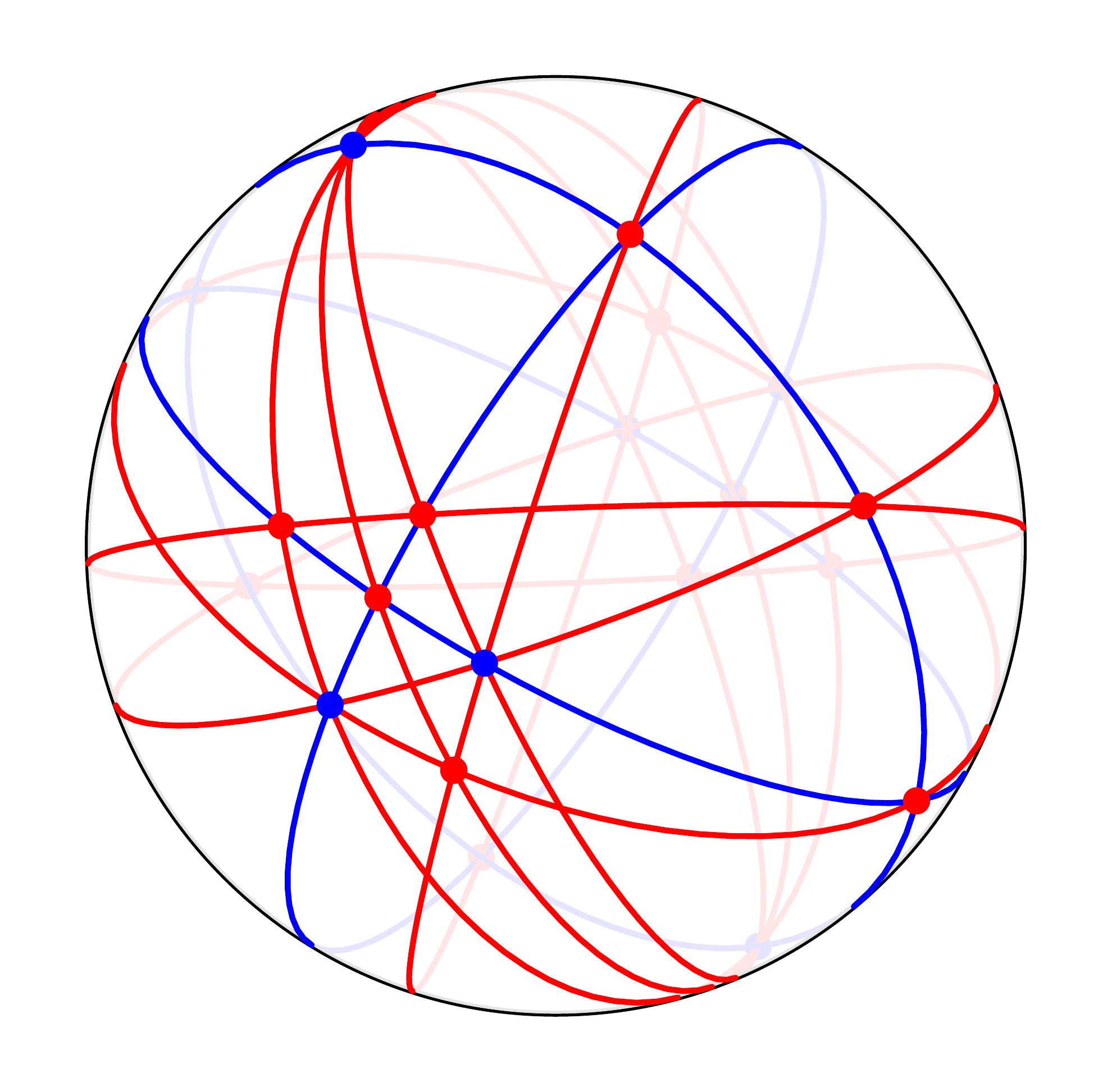}\hspace{-.5cm}
  \includegraphics[width=.35\textwidth]{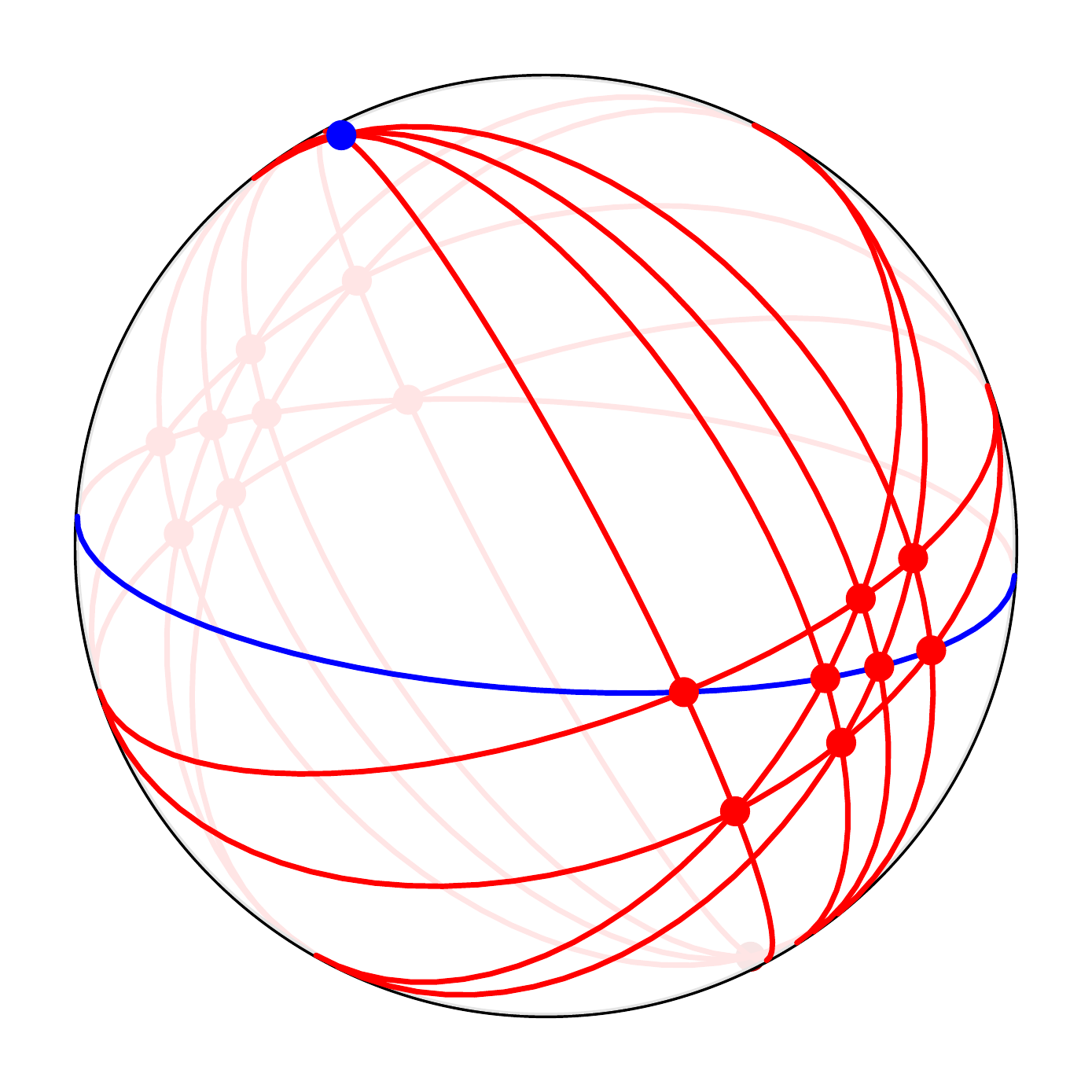}
  }
  \caption{Optimal $(n_{3|4})$ configurations, for~$n = 13, 12, 11, 10, 9$. They have respectively $48$, $42$, $37$, $33$, and $28$ point\,--\,line incidences. The $3$-valent elements are colored red while the $4$-valent elements are colored blue.}
  \label{fig:optimal}
\end{figure}

For~$n = 10, 11$ or~$12$, we obtain $(n_{3|4})$ configurations by removing suitable points and lines in our~$(13_{3|4})$ configuration. The resulting configurations are illustrated in Figure~\ref{fig:optimal}. (Note that for $n=10$, we even have two dual ways to suitably remove three points and three lines from our $(13_{3|4})$ configuration: either we remove three $3$-valent points and the three $4$-valent lines containing two of these points, or we remove three $3$-valent lines and the three $4$-valent points contained in two of these lines). Finally, for~$n=9$ we use the bottom rightmost $(9_{3|4})$ configuration of Figure~\ref{fig:optimal}.
\end{proof}

As a curiosity, we give another example of an optimal $(12_{3|4})$ configuration which contains Pappus' and Desargues' configurations simultaneously. See Figure~\ref{fig:optimal12PappusDesargues}.

\begin{figure}
  \centerline{\includegraphics[width=1\textwidth]{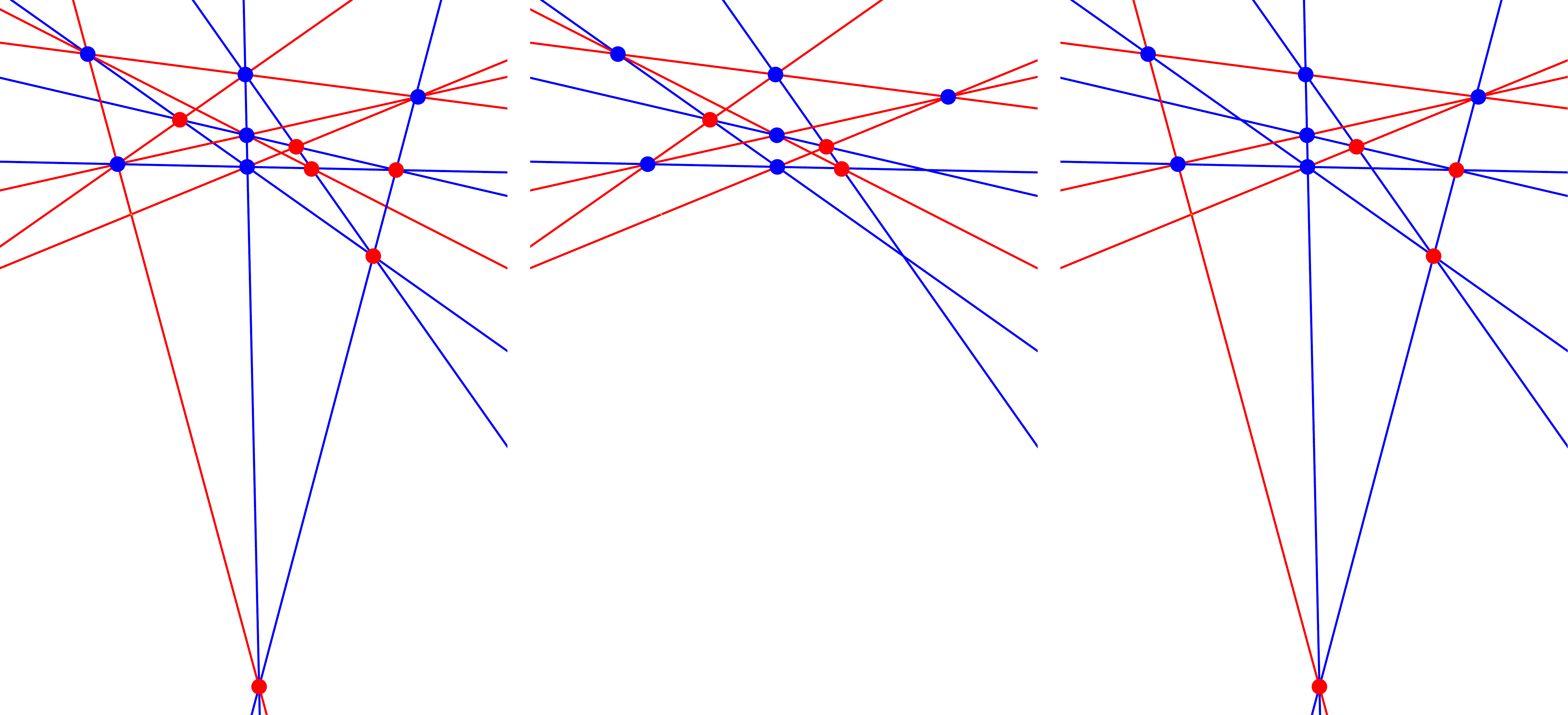}}
  \caption{An optimal $(12_{3|4})$ configuration (left) which contains simultaneously Pappus' configuration (middle) and Desargues' configuration (right). In the $(12_{3|4})$ configuration, the $3$-valent elements are colored red while the $4$-valent elements are colored blue. In the Pappus' and Desargues' subconfigurations, all elements are $3$-valent, but we keep the color to see the correspondence better.}
  \label{fig:optimal12PappusDesargues}
\end{figure}

Observe that optimal $(n_{3|4})$ configurations are given by $(n_4)$ configurations for large~$n$, and that the only remaining cases for optimal $(n_{3|4})$ configurations are for~$n = 14, 15, 16, 17, 19, 22, 23$, and~$26$. We have represented in Figure~\ref{fig:seemOptimal} some $(15_{3|4})$ and $(16_{3|4})$ configurations which we expect to be optimal, although they do not reach the theoretical upper bound of Corollary~\ref{coro:optimal}. Observe also that deleting the circle in Figure~\ref{fig:38_4}\,(left) and adding one line through two of the resulting $3$-valent points provides a $(19_{3|4})$ configuration with $74$ incidences, which is almost optimal since there is no $(19_4)$ configuration~\cite{BokowskiPilaud, BokowskiPilaud2}. To conclude, we thus leave the following question open.

\begin{question}
What are the optimal $(14_{3|4})$ configurations? Are the $(15_{3|4})$ and $(16_{3|4})$ configurations in Figure~\ref{fig:seemOptimal} optimal? Is there a $(19_{3|4})$ configuration with $75$ incidences.
\end{question}

\begin{figure}[h]
  \centerline{\includegraphics[width=.5\textwidth]{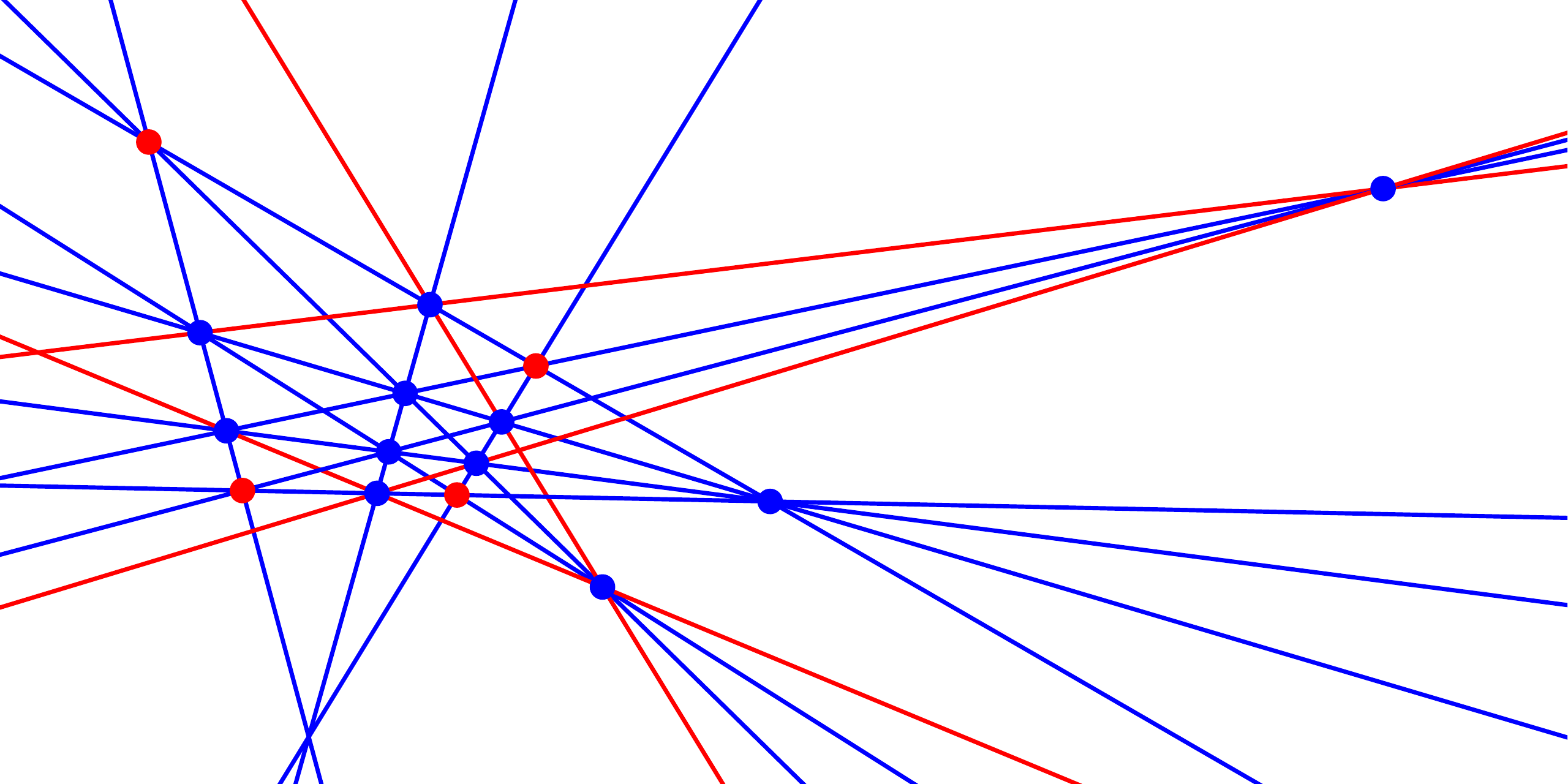}\hspace{.1cm}\includegraphics[width=.5\textwidth]{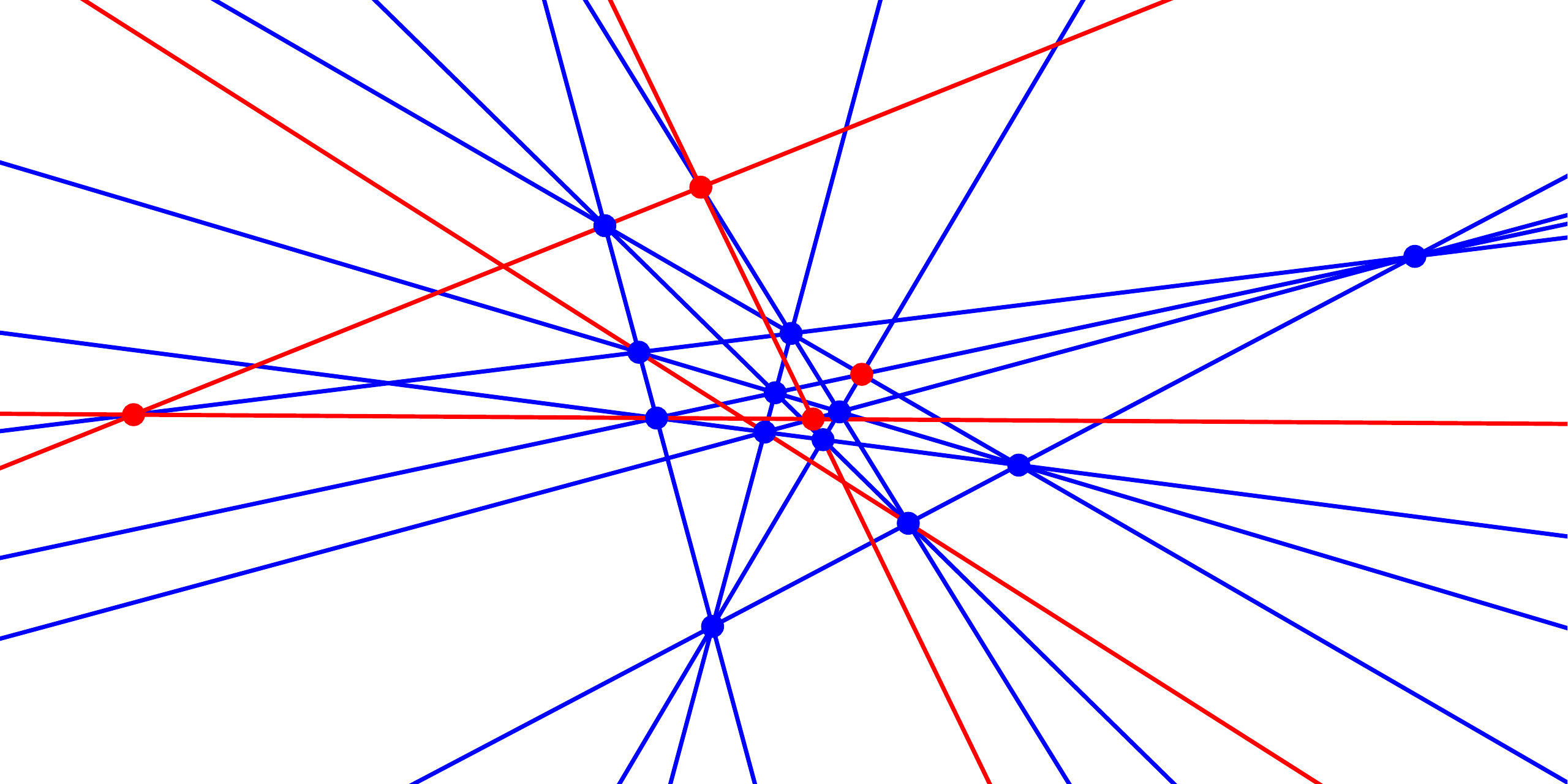}}
  \caption{Apparently optimal $(15_{3|4})$ and $(16_{3|4})$ configurations. They have $56$ and $60$ point\,--\,line incidences respectively. The $3$-valent elements are colored red while the $4$-valent elements are colored blue.}
  \label{fig:seemOptimal}
\end{figure}


\bibliographystyle{alpha}
\bibliography{QuasiConfigurationsAsBuildingBlocks.bib}

\end{document}